\documentclass{llncs}
\usepackage{makeidx}
\usepackage{graphicx}

\newcommand{\commentout}[1]{}

\newcommand{\QED}{\hbox{\hskip 4pt \vrule width 5pt height 6pt depth
    1.5pt\hskip 2pt}}

\newfont{\ssbb}{cmssbx10 scaled\magstep2}
\newfont{\ssb}{cmssbx10 scaled\magstep1}
\newfont{\stb}{cmssbx10}
\newfont{\teniu}{cmu10}

\begin{document}

\title{Dictionary Matching with One Gap\thanks{This research was supported by the Kabarnit Cyber consortium funded by the Chief Scientist in the Israeli Ministry of Economy under the Magnet Program.}}
\author{
Amihood Amir\inst{1}\fnmsep \inst{2}\fnmsep  \thanks{Partly supported by  NSF grant CCR-09-04581, ISF grant 347/09, and BSF grant 2008217.}
\and Avivit Levy\inst{3}
\and Ely Porat\inst{1}
\and B.~Riva Shalom\inst{3}}
\institute{Department of Computer Science, Bar-Ilan University,
Ramat-Gan 52900, Israel. \email{E-mail: \{amir, porately\}@cs.biu.ac.il} \and
Department of Computer Science, Johns Hopkins University,
Baltimore, MD 21218.
\and  Department of Software Engineering,
Shenkar College, Ramat-Gan 52526, Israel. \email {Email: \{avivitlevy,
rivash\}@shenkar.ac.il}}

\maketitle

\begin{abstract}
The dictionary matching with gaps problem is to preprocess a dictionary $D$ of $d$ gapped patterns $P_1,\ldots,P_d$ over alphabet $\Sigma$, where each gapped pattern $P_i$ is a sequence of subpatterns separated by bounded sequences of don't cares. Then, given a query text $T$ of length $n$ over alphabet $\Sigma$, the goal is to output all locations in $T$ in which a pattern $P_i\in D$, $1\leq i\leq d$, ends. There is a renewed current interest in the gapped matching problem stemming from cyber security. In this paper we solve the problem where all patterns in the dictionary have one gap with at least $\alpha$ and at most $\beta$ don't cares, where $\alpha$ and $\beta$ are given parameters. Specifically, we show that the dictionary matching with a single gap problem can be solved in either $O(d\log d + |D|)$ time and $O(d\log^{\varepsilon} d + |D|)$ space, and query time $O(n(\beta -\alpha )\log\log d \log ^2 \min \{ d, \log |D| \} + occ)$, where $occ$ is the number of patterns found, or preprocessing time and space: $O(d^2 + |D|)$, and query time $O(n(\beta -\alpha ) + occ)$, where $occ$ is the number of patterns found. As far as we know, this is the best solution for this setting of the problem, where many overlaps may exist in the dictionary.
\end{abstract}

\section{Introduction}\label{s:introduction}
Pattern matching has been historically one of the key areas of
computer science. It contributed many important algorithms and data
structures that made their way to textbooks, but its strength is that
it has been contributing to applied areas, from text searching and web
searching, through computational biology and to cyber security.
One of the important variants of pattern matching is pattern matching
with variable length gaps. The problem is formally defined below.

\begin{definition}\label{d:gappedpattern}
  {\em Gapped Pattern }\\
A Gapped Pattern is a pattern $P$ of the form  $ p_1\ \{ \alpha
_{1}, \beta_{1} \}\ p_{2}\  \{ \alpha _{2}, \beta_{2} \}
  \ldots \  \{ \alpha _{k-1}, \beta_{k-1} \} \  p_{k}$, \\
    where each subpattern $p_{j}$, $1 \leq j \leq k$ is a string over
alphabet $\Sigma$, and\\
  $ \{ \alpha _{j}, \beta_{j} \}$ refers to a
 sequence of at least $\alpha _{j}$  and at most
$\beta_{j}$ don't cares \\
between the subpatterns $P_{j}$ and $P_{ j+1}$.\\
\end{definition}

\begin{definition}\label{d:gappedmatch}
  {\em The Gapped Pattern Matching Problem}:\\
\begin{tabular}{ll}
 Input: &  A text $T$ of length $n$, and a gapped pattern $P$ over
 alphabet $\Sigma$\\
  Output: & All locations in $T$, where the pattern $P$ ends.\\
\end{tabular}
\end{definition}

The problem arose a few decades ago by real needs in computational
biology applications~\cite{mm-93,fp-08,nr-03,hpfb-99}.
For example, the PROSITE database~\cite{hpfb-99} supports queries for
proteins specified by gaps.

The problem has been well researched and many solutions proposed. The
first type of solutions~\cite{bt-09,myers-92,gr-04} consider the
problem as a special case of regular expression. The best time
achieved using this method is $O(n(B|\Sigma|+m))$, where $n$ is the
text length, $B=\sum_{i=1}^k \beta_i$ (the sum of the upper
bounds of the gaps), and $m=\sum_{i=1}^k p_i$ (the length of the
non-gapped part of the pattern).

Naturally, a direct solution of the gapped pattern matching problem
should have a better time complexity. Indeed, such a solution
exists~\cite{bt-10} whose time is essentially $O(nk)$.

A further improvement~\cite{mpvz-05,rilms:06,bgvw:12} analyses the
time as a function of $socc$, which is the number of times all
segments $p_i$ of the gapped pattern appear in the text. Clearly $socc
\leq nk$.
\commentout{
\subsubsection{Problem Definition.}
\begin{definition}\label{d:gappedpattern}
  {\em Gapped Pattern }\\
A Gapped Pattern is a pattern $P$ of the form $P_1\phi_{\alpha
_{1}}^{\beta_{1}}P_{2}\phi_{\alpha _{2}}^{\beta_{2}}\ldots \phi_{\alpha_{k}}^{\beta_{k}}P_{k+1}$, where each subpattern $P_{j}$, $1\leq j\leq k+1$ is a string over alphabet $\Sigma$, and $\phi_{\alpha_{j}}^{\beta_{j}}$, $1\leq j\leq k$, is a sequence of at least $\alpha_{j}$ and at most
$\beta_{j}$ don't cares.
\end{definition}
\begin{definition}\label{d:gappedDictionary}
  {\em The Dictionary Matching with Gaps ($DMG$) Problem} is the following:\\
\begin{tabular}{ll}
 Input: &  A text $T$ of length $n$ over alphabet $\Sigma$,\\
 & a dictionary $D$ of $d$ gapped patterns $P_1,\ldots,P_d$ over alphabet $\Sigma$.\\
  Output: & All locations in $T$ in which a pattern $P_i\in D$, $1\leq i\leq d$, ends.
\end{tabular}\\
For the gapped pattern $P_i\in D$ we denote:
$$P_i=P_{i,1}\phi_{\alpha_{1}}^{\beta_{1}}P_{i,2}\phi_{\alpha _{2}}^{\beta_{2}}\ldots \phi_{\alpha_{k}}^{\beta_{k}}P_{i,k+1}.$$
\end{definition}

\subsubsection{Previous Work.}\label{s:previouswork}
The problem of pattern matching of a single gapped pattern has been studied and several solutions were suggested.}

Rahman et al. \cite{rilms:06} suggest two algorithms for this problem. In the first, they build an Aho-Corasick \cite{AC-75} pattern matching machine from all the subpatterns and use it to go over the text. Validation of subpatterns  appearances  with respect to the limits of the gaps   is performed using binary search over previous subpatterns locations. Their second algorithm  uses a suffix array built over the text to locate occurrences of all subpatterns. For the validation of occurrences of subpattern, they use Van Emde Boas data structure~\cite{VanEmde}
containing ending positions of  previous occurrences. In order to report occurrences of the gapped pattern in both algorithms, they build a graph representing legal appearances of consecutive subpatterns. Traversing the graph yields all possible appearances of the pattern.

Their first  algorithm works in time $O(n + m + socc \log (\max_j gap_j))$ where $m$ is the length of the pattern (not including
the gaps), $socc $ is the total number of occurrences of the subpatterns in the text and $gap_j = \beta _j - \alpha _j $.
The time requirements of their second algorithm is $O( n + m + socc \log \log n)$ where $n$ is the length of the text, $m$ is the length of the pattern (not including the gaps) and $socc$ is the total number of occurrences of the subpatterns in the text.
The DFS traversal on the subpatterns occurrences graph, reporting all the occurrences is done in $O(k \cdot occ)$ where $occ$ is the
number of occurrences of the complete pattern $P$ in the text.

Bille et al.~\cite{bgvw:12} also consider string matching with variable length gaps. They present an algorithm using sorted lists of disjoint intervals, after traversing the text with Aho-Corasick automaton. Their time complexity is $O(n\log k+m+socc)$ and space $O(m+A)$, where $A$ is the sum of the lower bounds of the lengths of the gaps in the pattern $P$ and $socc$ is the total number of occurrences of the substrings in $P$ within $T$.

Kucherov and Rusinowitch \cite{KR:97} and Zhang et al.~\cite{ZZH:10} solved the problem of matching a set of patterns with variable length of don't cares. They considered the question of determining whether one of the patterns of the set matches the text and report a leftmost occurrence of a pattern if there exists one. The algorithm of~\cite{KR:97} has run time of $O((|t| + |D|) \log |P|)$, where $ |D|$ is the total length of keywords in every pattern of the dictionary $D$. The algorithm of~\cite{ZZH:10} takes $O((|t| + dk )\log dis/ \log  \log dis)$ time, where $dk$ is the total number of keywords in every pattern of $P$, and $dis$ is the number of distinct keywords in $D$.

There is a renewed current interest in the gapped matching problem
stemming from a crucial modern concern - cyber security. Network
intrusion detection systems perform protocol analysis, content
searching and content matching, in order to detect harmful
software. Such malware may appear on several packets, and thus the
need for gapped matching~\cite{sl-08}. However, the problem becomes
more complex since there is a large list of such gapped patterns that
all need to be detected. This list is called a {\em dictionary}. Dictionary matching has been amply researched in
computer science (see e.g.~\cite{AC-75,\commentout{a:f:IPL:92,}AFILS-92-journal,\commentout{Idury-Schaeffer-92:cpm,}amir:f:g:g:park:93,BG:96,aklllr00,a:calinescu:00,cgl-04}). We are concerned with a new dictionary matching paradigm - {\em dictionary matching with gaps}. Formally:

\begin{definition}\label{d:gappedDictionary}
  {\em The Dictionary Matching with gaps ($DMG$) Problem}:\\
\begin{tabular}{ll}
 Input: &  A text $T$ of length $n$ over alphabet
$\Sigma$ and a dictionary $D$ over \\
&  alphabet $\Sigma$ consisting of $d$
gapped patterns  $P_1, \ldots, P_d$.\\
  Output: & All locations in $T$, where a
pattern $P_i$, for $1 \leq i \leq d$ ends. \\
\end{tabular}
\end{definition}

The $DMG$ problem has not been sufficiently studied yet. Haapasalo et al.~\cite{hsss:11} give an on-line algorithm for the general problem. Their algorithm is based on locating ``keywords'' of the patterns in the input text, that is, maximal substrings of the patterns that contain only input characters. Matches of prefixes of patterns are collected from the keyword matches, and when a prefix constituting a complete pattern is found, a match is reported. In collecting these partial matches they avoid locating those keyword occurrences that cannot participate in any prefix of a pattern found thus far. Their experiments show that this algorithm scales up well, when the number of patterns increases. They report at most one occurrence for each pattern at each text position. The time required for their algorithm is $O(n\cdot SUF + occ\cdot PREF)$, where $n$ is the size of the text, $SUF$ is the maximal number of suffixes of a keyword that are also keywords, and $PREF$ denotes the number of occurrences in the text of pattern prefixes ending with a keyword.

Nevertheless, more research on this problem is needed. First, in many applications it is necessary to report all patterns appearances. Moreover, as far as we know \cite{VerInt}, these Aho-Corasick automaton based methods fail when applied to real data security, which contain many overlaps in the dictionary patterns, due to overhead in the computation when run on several ports in parallel. Therefore, other methods should be developed and combined with the existing ones in order to design efficient practical solutions.

\subsubsection{Results.}
In this paper, we indeed suggest other directions for solving the problem. We focus on the $DMG$ problem where the gapped patterns in the dictionary $D$ have only a single gap, i.e.,\ we consider the case of $k=1$ implying each pattern $P_i$ consists of two subpatterns $P_{i,1}, P_{i,2}$. In addition, we consider the same gaps limits, $\alpha$ and $\beta$,
apply to all patterns in $P_i\in D$, $1\leq i\leq d$. We prove:
\begin{theorem}\label{t:single}
The dictionary matching with a single gap problem can be
solved in:
\begin{enumerate}
  \item Preprocessing time: $O(d\log d + |D|)$.\\
  Space:  $O(d\log ^{\varepsilon}d + |D|)$, for arbitrary small $\varepsilon$.\\
  Query time: $O(n(\beta -\alpha )\log\log d \log ^2 \min \{ d, \log |D| \} + occ)$, where $occ$ is the number of patterns found.
  \item Preprocessing time: $O(d^2 + |D|)$.\\
  Space: $O(d^2 + |D|)$.\\
Query time: $O(n(\beta -\alpha ) + occ)$, where $occ$ is the number of patterns found.
\end{enumerate}
\end{theorem}
Note that, $|D|$ is the sum of lengths of all patterns in the dictionary, \emph{not including} the gaps sizes.
\commentout{
\begin{theorem}\label{t:kgaps}
The dictionary matching with $k$ gaps problem, where $k>1$, can be
solved in:\\
Preprocessing time: $O(kd\log d + |D|)$.\\
Query time: $O(nk(\beta -\alpha )\sqrt {\log d} \log ^2 \min \{ d, \log |D| \} + occ$, where $occ$ is the number of patterns found.
\end{theorem}
}

The paper is organized as follows. In Sect.~\ref{s:trees} we describe our basic method based on suffix trees and prove the first part of Theorem~\ref{t:single}. In Sect.~\ref{s:lookup} we describe how this algorithm query time can be improved while doing more work in the preprocessing time and prove the second part of Theorem~\ref{t:single}. We also discuss efficient implementations of both algorithms using text splitting in Subsect.~\ref{ss:split}. Section~\ref{s:open} concludes the paper and poses some open problems.


\section{Bidirectional Suffix Trees Algorithm}\label{s:trees}
The basic observation used by our algorithm is that if a gapped pattern $P_i$ appears in $T$, then searching to the left from the start position of the gap we should find the reverse of the prefix $P_{i,1}$, and searching to the right from the end position of the gap we should find the suffix $P_{i,2}$. A similar observation was used by Amir et al.~\cite{aklllr00} to solve the dictionary matching with one mismatch problem. Their problem
is different from the $DMG$ problem, since they consider a single
mismatch while in our problem a gap may consist of several symbols. Moreover, a mismatch symbol
appears both in the dictionary pattern and in the text, while in the
$DMG$ problem the gap implies skipping symbols only in the text. Nevertheless, we show that their idea can also be adopted to the solve the $DMG$ problem.

Amir et al.~\cite{aklllr00} use two suffix trees: one for the
concatenation of the dictionary and the other for the reverse of
the concatenation of the dictionary. Combining this with set intersection on
tree paths, they solved the dictionary matching with one mismatch problem in time
$O(n\log^{2.5}|D|+ occ)$ where $n$ is the length of the text, $|D|$
is the sum of the lengths of the dictionary patterns, and $occ$ is
the number of occurrences of patterns in the text. Their
preprocessing requires $O(|D|log |D|)$. We use the idea to design an algorithm to the $DMG$ problem.

A naive method is to consider matching the prefixes $P_{i,1}$ for all $1\leq i\leq d$, and then
look for the suffixes subpatterns $P_{i,2}$, $1\leq i\leq d$, after the
appropriate gap and intersect the occurrences to report the dictionary patterns matchings.
However, as some of the patterns may share subpatterns and some
subpatterns may include other subpatterns, there may be several
distinct subpatterns occurring at the same text location, each of
different length. Therefore, we need to search for the suffixes $P_{i,2}$, $1\leq i\leq d$, after several gaps, each beginning at the end
of a matched prefix $P_{i,1}$, $1\leq i\leq d$. To avoid multiple searches we search
all subpatterns $P_{i,1}$, $1\leq i\leq d$ that \emph{end} at a certain location. Note that in order to find all subpatterns
ending at a certain location of the text we need to look for them backwards and find their reverse $P_{i,1}^R$.\commentout{ An example of matching a pattern's suffix and prefix appears in Figure \ref{f:gapped matching}.

\begin{figure}[h]\label{f:gapped matching}
\begin{center}
$P_1 = $ $a\ c\ \{2, 4\} \ d\ d$\\
$P_2 = $ $a\ b\ \{2, 4\} \ c\ d$

 \vspace{0.3cm}
  $ T =  c\ \ d\ \ e\ \ f \ \ \overleftarrow{  a\ \ b\ }\ e\ \ b\ \ \overrightarrow{ c\ \ d\ } \ a\ \ c\
  $ \\
 \small{ $\ \  \qquad 1\ \  2\ \  3\ \  4\  \ 5\ \ \ 6\ \ 7\  \ 8\ \ 9\  10 \ 11 \
 12$}
  \end{center}
\caption{Matching $P_{2,2}$ starting from $t_9$ skipping $\alpha = 2$ symbols backwards and matching $P_{2,1}^R$ to a reverse
prefix of $T$ ending at location $t_6$.}
\end{figure}}

In the preprocessing stage we concatenate the subpatterns $P_{i,2}$, $1\leq i\leq d$ of the dictionary separated by the symbol $
\$ \notin \Sigma$ to form a single string $S$. We repeat the procedure
for the subpatterns $P_{i,1}$, $1\leq i\leq d$, to form a single string $F$. We construct a suffix tree $T_S$ of
the string $S$, and another suffix tree $T_{F^R}$ of the string $F^R$, which is the reverse of the string $F$.

We then traverse the text by inserting suffixes of the text to the
$T_S$ suffix tree. When we pass the node of $T_S$ for which the path from the root is labeled $P_{i,2}$, it
implies that this subpattern occurs in the text starting from the
beginning of the text suffix. We then should find whether $P_{i,1}$
also appears in the text within the appropriate gap from the
occurrences of the $P_{i,2}$ subpatterns. To this end, we go backward
in the text skipping as many locations as the gap requires and
inserting the reversed prefix of the text to $T_{F^R}$. If a
node representing $P_{i,1}$ is encountered, we can output that
$P_i$ appears in the text.

Note that several dictionary subpatterns representative nodes
may be encountered while traversing the trees. Therefore, we should report the intersection between the subpatterns found from
each traversal. We do that by efficient intersection of labels on
tree paths, as done in Amir et al.~\cite{aklllr00}, using range queries on a grid.
However, since some patterns may share subpatterns, we do not label the
tree nodes by the original subpatterns they represent, as done in \cite{aklllr00}. Instead, we mark the nodes
representing subpatterns numerically in a certain order, hereafter discussed, regardless to the origin of the subpatterns ending at
those nodes.

In order to be able to trace the identity of the
patterns from the nodes marking, we keep two arrays $A_F$ and $A_S$, both of size
$d$. The arrays contain linked lists that identify when subpatterns are shared among several dictionary patterns. These arrays are filled as follows: $A_F[g] = i$ if $P_{i,1}^R$ is represented by node labelled $g$ in $T_{F^R}$ and $A_S[h] =i$ if $P_{i,2}$ is represented by the node
labelled by $h$ in $T_S$. In addition, $A_F[g]$ is linked to $h$ and vice versa, if $g, h $ represent two subpatterns
of the same pattern of the dictionary.

Every pattern $P_i$ is represented as a point $i$ on a grid of size $d \times d$,
denoted by $< g, A_F[g].link >$, that is, the $x$-coordinate is the mark of the node representing $P_{i,1}^R$ in $T_{F^R}$, and the $y$-coordinate is the mark of the node representing $P_{i,2}$ in $T_S$. Now, if we mark the nodes representing the end of subpatterns so
that the marks on a path are consecutive numbers, then the problem of
intersection of labels on tree paths can be reduced to range queries on a grid in the following way.
Let the first and last mark on the relevant path in $T_{F^R}$ be $g, g'$ and, similarly, on the path in $T_S$ let the first and last mark be $h,h'$. Thus, points $<x, y>$ on the grid where $g\leq x \leq g'$ and $h \leq y \leq h'$, represent patterns in the dictionary for which both subpatterns appear at the current check. A range query can be solved using the algorithm of \cite{CLP-11}.

We use the decomposition of a tree into vertical path for the nodes
marking, suggested by \cite{aklllr00}, though we use it
differently. A vertical path is defined as follows.
\begin{definition}\cite{aklllr00} \label{d:verticalpath}
A \emph{vertical path} of a tree is a tree path (possibly
consisting of a single node) where no two  nodes on the path are
of the same height.
\end{definition}
After performing the decomposition, we can traverse the vertical
paths and mark by consecutive order all tree nodes representing
the end of a certain subpattern appearing in $T_S$ or its reverse
appearing in $T_{F^R}$. Since some subpatterns may be shared by
several dictionary patterns, there are \emph{at most} $d$ marked nodes at each of
the suffixes trees.

Note that due to the definition of vertical path there may be
several vertical paths that have a non empty intersection with the
unique path from the root to a specific node. Hence, when considering the
intersection of marked nodes on the path from the root till a
certain node marked by $g$ in $T_{F^R}$ and the marked nodes
on the path from the root till a node marked by $h$ in
$T_S$, we actually need to check the intersection of all
vertical paths that are included in the path from the root to
$g$ with all vertical paths that are included in the path from
the root to $h$.

The algorithm appears in Figure 1\commentout{ \ref{f:secondalg1gap}}.

\begin{figure}[h]\label{f:secondalg1gap}
\begin{tabular}{ll}
 \hline
 & \\
& \textsc{Single\_ Gap\_ Dictionary} ($T, D$ )\\
\hline
\hline
& \textbf{Preprocessing}:\\
\hspace{0.1cm} 1 &\hspace{0.2cm}$F = P_{1,1}\$  P_{2,1}\$ \cdots P_{d,1}$.\\
\hspace{0.1cm} 2 &\hspace{0.2cm}$S = P_{1,2}\$  P_{2,2}\$ \cdots P_{d,2}$.\\
\hspace{0.1cm} 3 &\hspace{0.2cm}$T_{S} \leftarrow $ a suffix tree of $S$.\\
\hspace{0.1cm} 4 &\hspace{0.2cm}$T_{F^R} \leftarrow $ a suffix tree of the $F^R$.\\
\hspace{0.1cm} 5 &\hspace{0.2cm}\textbf{For} every edge $(u,v) \in \{T_{S}, T_{F^R}\}$ with label $y\$ z$, where $y,z \in \Sigma^*$\\
\hspace{0.1cm} 6 & \hspace{0.4cm} Break $(u,v)$ into $(u,w)$ and $(w,v)$ labelling $(u,w)$ with $y$ and $(w,v)$ with $ \$ z$.\\
\hspace{0.1cm} 7 &\hspace{0.2cm}Decompose $T_{F^R}$ into vertical paths.\\
\hspace{0.1cm} 8 &\hspace{0.2cm}Mark the nodes representing $P_{i,1}$ on the vertical paths of $T_{F^R}$.\\
\hspace{0.1cm} 9 &\hspace{0.2cm}Decompose $T_S$ into vertical paths.\\
\hspace{0.1cm} 10 &\hspace{0.2cm}Mark the nodes representing $P_{i,2}$ on the vertical paths of $T_S$.\\
\hspace{0.1cm} 11 &\hspace{0.2cm}Preprocess the points according to the patterns for range queries.\\
\hspace{0.1cm}  &  \\
& \textbf{Query}:\\
\hspace{0.1cm} 12 &\hspace{0.2cm}\textbf{For} $\ell$ = $\min _i\{ |P_{i,1}|\}  + \alpha $ to $n$ \\
\hspace{0.1cm} 13 &\hspace{0.4cm} Insert  $t_{\ell}t_{\ell+1}\ldots t_n$ to $T_S$.\\
\hspace{0.1cm} 14 &\hspace{0.4cm} $h \leftarrow$ node in $T_S$ representing suffix $t_{\ell}t_{\ell+1}\ldots t_n$.\\
\hspace{0.1cm} 15 &\hspace{0.4cm}\textbf{For} $f =\ell - \alpha - 1$ to  $\ell - \beta - 1$\\
\hspace{0.1cm} 16 &\hspace{0.6cm} Insert $t_{f}t_{f-1}\ldots t_1$ to $T_{F^R}$.\\
\hspace{0.1cm} 17 &\hspace{0.6cm} $g \leftarrow $ node in $T_{F^R}$ representing $t_{f}t_{f-1}\ldots t_1$.\\
\hspace{0.1cm} 18 &\hspace{0.6cm} \textbf{For} every vertical path on the the path $p$ from the root to $h$ \\
\hspace{0.1cm} 19 &\hspace{0.8cm} \textbf{For} every vertical path $p'$ on the path from the root to $g$ \\
\hspace{0.1cm} 20 &\hspace{1cm} Perform a range query on a grid with the first and last marks of $p$ and $p'$ \\
\hspace{0.1cm} 21 &\hspace{1cm} Report appearance for every $P_i$ where point $i$ appears in the specified range.\\
\hline
\end{tabular}
\caption{Dictionary matching with a single gap algorithm, intersection is computed by range queries on a grid.}
\end{figure}

\begin{lemma}\label{l:intgrid}
The intersection between the subpatterns appearing at location
$t_{\ell}$ and the reversed subpatterns ending at $t_{\ell-gap-1}$ can be
computed in time $O(occ + \log \log d  \log ^2 \min \{ d, \log |D| \})$, where $occ$ is the number of patterns found. The preprocessing requires $O(|D| + d\log d)$ time and $O(|D| + d\log^{\varepsilon} d)$ space, for arbitrary small $\varepsilon$.
\end{lemma}

At each of the $O(n)$ relevant locations of the text, the
algorithm inserts the current suffix of the text to $T_S$ using Weiner's
algorithm \cite{W-73}. For each of the prefixes defined by all $\beta - \alpha + 1$ possible specific gaps we
insert its reverse to $T_{F^R}$. As explained in~\cite{aklllr00}, the navigation on the suffix tree and reverse suffix tree can be done in amortized $O(1)$ time per character insertion. Note that each character is inserted to $T_S$ once and to $T_{F^R}$ $O(\beta-\alpha)$ times. This concludes the proof of the first part of Theorem~\ref{t:single}.

\section{Intersection by Lookup Table}\label{s:lookup}
If a very fast query time is crucial and we are willing to pay in
preprocessing time, we can solve the problem of intersection
between the appearances of subpatterns on the paths of $
T_{F^R}, T_S$ using a lookup table.

The $inter$ table is of size $d \times d$,  where $inter[g,h]$
refers to the set of all indices $i$ of patterns $P_i$ such that
$P_{i,1}^R$ appears on the path from the root of $T_{F^R}$ till
the node marked by $g$ and $P_{i,2}$ appears on the path from the
root of $T_S$ till the node marked by $h$. We fill the table using
dynamic programming procedure.  Consequentially, labelling the nodes
representing subpatterns, can be done by any numbering system,
guaranteeing that nodes closer to the root are labelled by smaller
numbers than nodes farther from the root, such as the BFS order.

Saving at every entry all the relevant pattern indices causes
redundancy in case subpatterns include others as their prefix or
suffix. In order to save every possible occurrence
only once, we save pattern index $i$ only  at entry $inter[g, h]$
where $g, h$ are the nodes respectively representing both subpatterns of $P_i$
 in the suffix trees. Note that at most one index can
be saved at $inter[g, h].index$ as two patterns are bound to
differ by at least one subpattern. The filling of these $d$ fields
is done in the preprocessing.

  We hereafter prove that besides the $index$ field,  merely 3 links are required for every
$inter[g,h]$ :
\begin{enumerate}
 \item A link to $inter[g', h]$ in case node $h$ represents subpattern $P_{i, 2}$
and $g'$ is the maximal labelled ancestor of node $g$, representing
$P_{i, 1}$. We call this link an $up$ link.
\item A link to cell $inter[g, h']$ in case
 node $g$ represents the subpattern $P_{i, 1}$ and $h'$ is the
maximal labelled ancestor of node $h$
 representing
$P_{i, 2}$. We call this link a $left$ link.
\item A link to  cell $[prev^*(g), prev^*(h)]$, where $prev^*(g)$
and $prev^*(h)$ are  the closest marked ancestors of the nodes
marked by $g$ and $h$ where  $inter[prev^*(g),$ $ prev^*(h)]$ has a
pattern index or non null $up$ or $left$ link.
 \end{enumerate}

 The recursive rule for
constructing the lookup table is described in the following lemma.
\begin{lemma}{\textbf{The Recursive Rule }}\label{l:interrecursive}\\
Let $prev(x)$ be the maximal labelled ancestor of the node labelled by $x$.
\begin{displaymath}
inter[g,h].up= \left\{ \begin{array}{l}
[prev(g), h] $ \qquad \qquad \quad \textbf{if}  \ $ inter[prev(g), h].index \neq null\\
inter[prev(g),h].up  \quad \quad $\textbf{otherwise}$ \\
   \end{array}
  \right.
\end{displaymath}
\begin{displaymath}
inter[g,h].left= \left\{ \begin{array}{l}
[g, prev( h)] $ \qquad \qquad \quad \textbf{if}  \ $ inter[g,prev(h)].index \neq null\\
inter[g,prev(h)].left  \quad \quad $\textbf{otherwise}$ \\
   \end{array}
  \right.
\end{displaymath}
\begin{displaymath}
inter[g,h].prev= \left\{ \begin{array}{l}
 [prev(g),prev(g)] $ \ \textbf{if}  \ $ inter[prev(g),prev(h)].index \neq null\\
 \quad \qquad \qquad \qquad \qquad $ or\ $inter[prev(g),prev(h)].up \neq null\\
 \quad \qquad \qquad \qquad \qquad $or\ $inter[prev(g),prev(h)].left \neq null\\
       inter[prev(g),prev(g)].prev   \quad $\textbf{otherwise}$ \\
   \end{array}
  \right.
  \end{displaymath}
\end{lemma}
\begin{proof}

Every $inter[g,h]$ entry for $1\leq g, h \leq d$, has to contain links to all entries  containing indices  of patterns whose first
subpattern is represented by node $g$ or  its ancestors in $T_{F^R}$ and its second subpattern is represented by node $h$ or its ancestors in $T_S$.
Assume, without loss of generality, that the marks on the
path from the root of $T_{F^R}$ till the node marked by $g$ are
$g'_1, g'_2,...g'_a, g$ and the marks on the path from the root of
$T_S$ till the node marked by $h$ are $h'_1, h'_2,...h'_b, h$.
There are four cases of relevant entries, for each we prove the correctness of the recursive rule.
\begin{enumerate}
\item Case 1: In case the pair of labels $<g, h>$ represent pattern $P_i$, then assygn $inter[g,h].index$ by $i$ in the preprocess. No recursion is required.

\item Case 2: Some $\{ g'_x \} $ represent the reverse of first subpattern of  some dictionary patterns (as theses subpatterns include  others as their suffixes),
and these patterns  share the second subpattern where $h$ represent this second subpattern.
In  such a case $inter[g,h]$ should be linked to all entries $\{inter[g'_x, h] \}$. Nevertheless, saving a link to the
ancestor node  with maximal label, is sufficient, since
all other relevant patterns can be reached by recursively
following  $up$ links starting from $inter[g'_x, h]$.
Therefore,  we consider $prev(g)$ as the  maximal labeled ancestor and if  $inter[prev(g), h]$ includes an index, we assign $up$ with
$[prev(g), h]$. Otherwise we are seeking the same ancestor
$inter[prev(g),h]$ looked for, for its $up$ link, hence we assign $up$ with $inter[prev(g),
h].up$.

\item  Case 3: Some $\{ h'_y \} $ represent the second subpattern of  some dictionary patterns (as theses subpatterns include  others as their prefixes),
and these patterns  share the first subpattern where $g$ represent the reverse of this first subpattern. Due to similar arguments
we assign the $left$ link either with $[g, prev(h)]$   or with $inter[g,
prev(h)].left$.

\item Case 4: Some ancestors of  nodes $g$ and $h$ represent both subpatterns of dictionary patterns. Note that it must be ancestors to both nodes as the previous  cases dealt with representations of patterns using the nodes $g$ or $h$ themselves. We need to enable $inter[g, h]$  to follow all such entries, to this end we look for the closest such ancestors. We check whether $inter[prev(g), prev(h)]$ includes an index
or an $up$ or $left$ link. If it does, link the $prev$ link to
$[prev(g), prev(h)]$. If it does not, it means that no pattern is represented  either by node $prev(g)$ and a node on the path from the
 root of $T_S$ till node $prev(h)$  or by a node on the path from the root of $T_{S^R}$ till node $prev(g)$ and the node $prev(h)$. Therefore we can step back in both pathes of the trees to  $prev(prev(g))$ and $prev(prev(h))$. Such a scenario can repeat, yet $inter[prev(g), prev(h)].prev$  which is $[prev^*(g), prev^*(h)]$,  was already computed, due to the numbering system, and the relevant information  is bound to appear at that entry, so we can follow it by assigning $prev = inter[prev(g), prev(h)].prev$.

 \end{enumerate}
 \QED
\end{proof}

\paragraph{\textbf{Example.}}An example of using the recursive rule, filling the $inter$ table
can be seen in Figure~\ref{f:lookup} for the  trees depicted in
Figure~\ref{f:trees} and the following dictionary.
 $P_1 = < 3, 9 >$, $P_2= <3,5>$, $P_3 = < 2,9 >$, $P_4=< 2,7 >$, $P_5= < 1,2 >$, $P_6 = < 1,5 >$, $P_7 = < 4,1 >$, $P_8= < 3,4 >$, $P_9 = < 2,3 >$, $P_{10} = < 4,8 >$.
Note the dashed $prev$ arrow from $inter[3,8]$ to $inter[1,5]$, due to the lack of information in $inter[2,6]$, which is the entry of the immediate ancestors of $[3,8]$.
\begin{figure}
  \begin{center}
\includegraphics[scale=0.35]{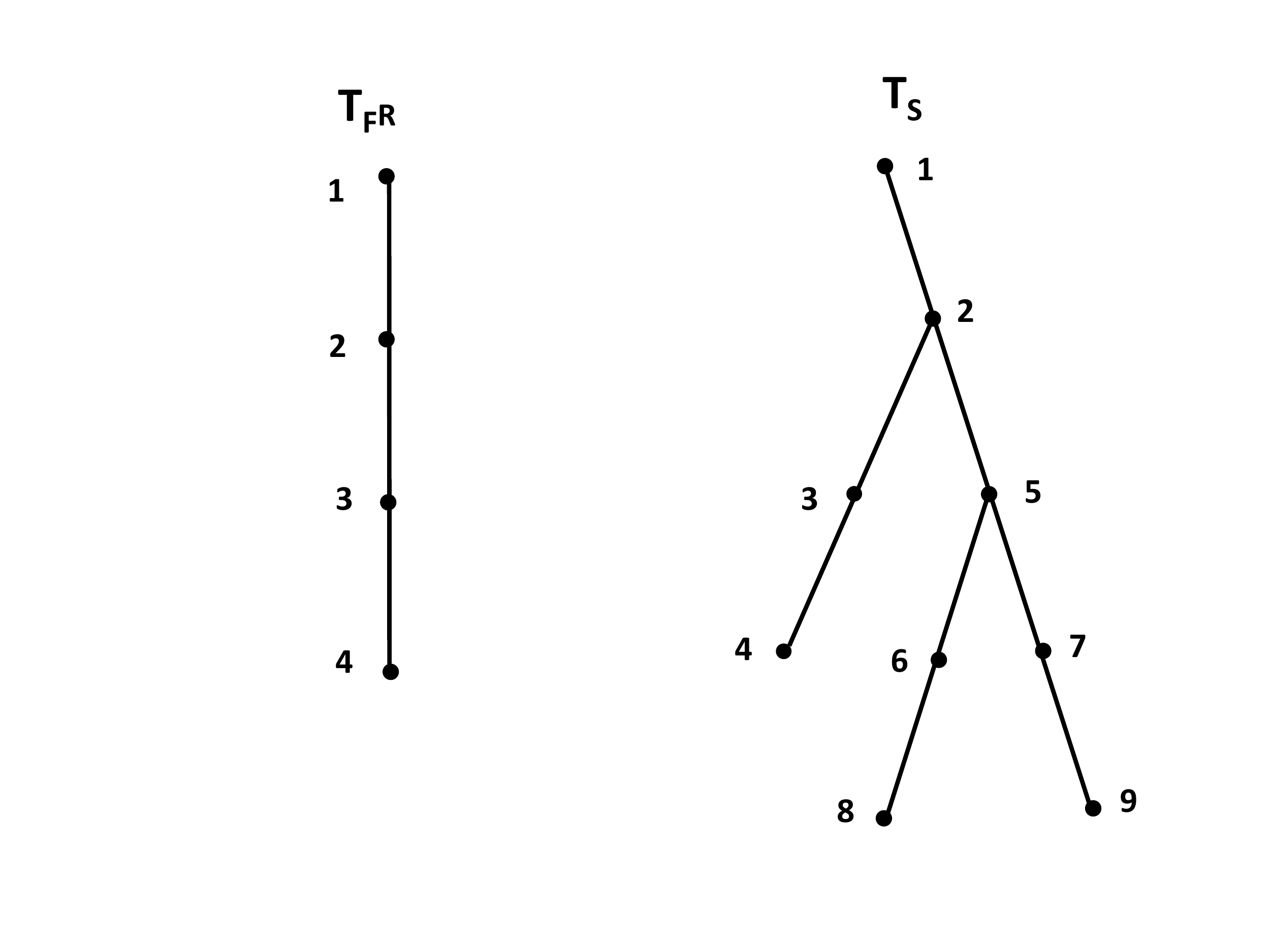}\\
  \caption{Two suffix trees are shown, where
  the nodes representing subpatterns are marked by numerical
  labels.}\label{f:trees}
\end{center}
\end{figure}
\begin{figure}
  \begin{center}
 \includegraphics[height=3in,width=5in]{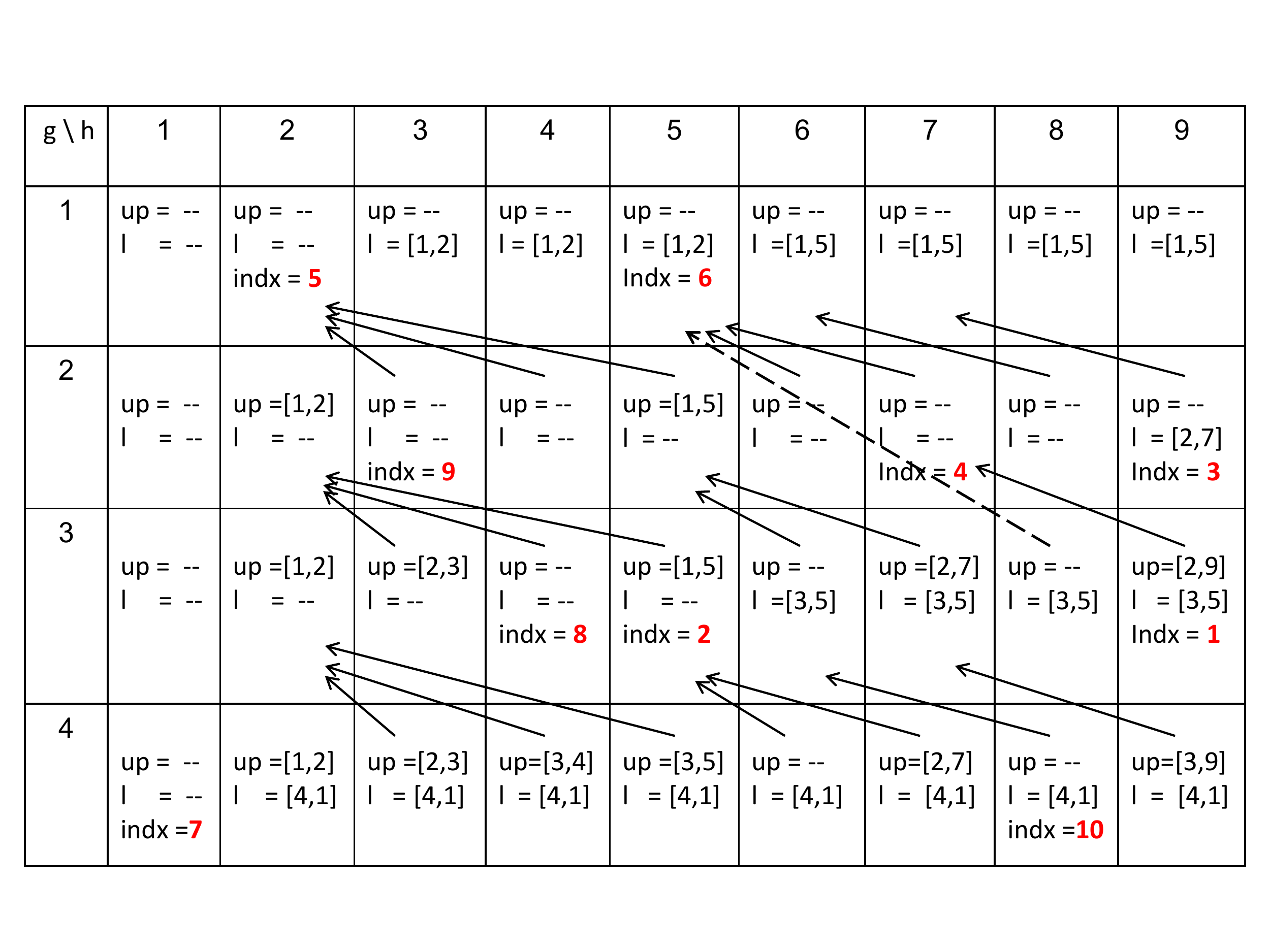}
  \caption{The lookup table built according to the trees depicted
  in Figure~\ref{f:trees}. The arrows represent the $prev$ links.}\label{f:lookup}
\end{center}
\end{figure}

Lemma~\ref{l:preLook} gives the preprocessing time and space
guarantee.
\begin{lemma}\label{l:preLook}
Preprocessing to build the $inter$ table r
equires $O(|D| + d ^ 2
)$ time.
\end{lemma}
\begin{proof}
The preprocess requires  labelling  both suffix trees
in BFS order all in $O(|D|)$. Filling $d$ $inter[g,h]$ entries
with the index of the pattern the nodes $g,h$ represent can be
done in at most $O(d^2)$.

 Filling each of the $d^2$
entries of the table $inter[g, h]$ can be performed in $O(1)$
due to Lemma \ref{l:interrecursive}.
\QED
\end{proof}

\paragraph{\textbf{Answering LookUp Queries.}}

A query of a node $g$ from $T_{F^R}$ and node $h$ from $T_S$ is answered by  consulting
 entry $inter[g,h]$. We  output $inter[g,h].index$ if exists which means  there is a pattern whose
first subpattern is represented by node $g$ and its second aubpattern is represented by node $h$.
In order to report all relevant patterns, that their subpatterns are represented by $g$ or its ancestors and by $h$ or
its ancestors in the suffix trees, we  follow the links saved at the current entry, as detailed in the procedure appearing at Figure~\ref{f:query}.

\begin{figure}[h]
\begin{center}
\begin{tabular}{rl}
 \hline
 & \\
& \textsc{LookupQuery}( g, h ) \\
\hline
\hline
& \\
\hspace{0.6cm} 1 & \hspace{0.8cm}\textbf{If} $inter[g, h].index \neq null  \qquad \qquad \quad $  \\
\hspace{0.6cm} 2 & \hspace{1.3cm} Output $inter[g, h].index$ \\
 & \\
\hspace{0.6cm} 3 & \hspace{0.8cm}\textbf{If} $inter[g, h].prev \neq null$\\
\hspace{0.6cm} 4 & \hspace{1.3cm} Let $[g', h'] \leftarrow  inter[g,h].prev $.\\
\hspace{0.6cm} 5 & \hspace{1.3cm} \textsc{LookupQuery}( g', h' )  \\
\hspace{0.6cm}  &  \\
\hspace{0.6cm} 6 &  \hspace{0.8cm}Let $ G \leftarrow g $.\\
\hspace{0.6cm} 7 & \hspace{0.8cm}\textbf{ While}  $( inter[g,h].up \neq null$ ).\\
\hspace{0.6cm} 8 & \hspace{1.3cm} Let $[g', h] \leftarrow inter[g,h].up $.\\
\hspace{0.6cm} 9 & \hspace{1.3cm} Output $inter[g', h].index$ \\
\hspace{0.6cm} 10 & \hspace{1.3cm} $ g \leftarrow g'$.\\
& \\
\hspace{0.6cm} 11 & \hspace{0.8cm}\textbf{ While}  $( inter[G,h].left \neq null $).\\
\hspace{0.6cm} 12 & \hspace{1.3cm} Let $[G, h'] \leftarrow inter[G,h].left $.\\
\hspace{0.6cm} 13 & \hspace{1.3cm} Output $inter[G, h'].index$ \\
\hspace{0.6cm} 14 & \hspace{1.3cm} $ h \leftarrow h'$.\\
\hline
\end{tabular}
\end{center}
\caption{The Lookup Query Procedure } \label{f:query}
\end{figure}

Lemma~\ref{l:queryLookTime} gives the query time guarantee.
\begin{lemma}\label{l:queryLookTime}
Using the $inter$ table, the intersection between the subpatterns
appearing at location $t_{\ell}$ and the reversed subpatterns
ending at $t_{\ell-gap-1}$ can be computed in time $O(occ)$, where
$occ$ is the number of patterns found.
\end{lemma}

\begin{proof}
  The query procedure is based on following links and reporting indices found.
Every step of following an $up$ or $left$  link  implies that another pattern is reported, as
those links connect two subpatterns including one another, where
both should be reported.  The $prev$ link either directs us to an
entry including a pattern index, needs to be reported or it
directs us to an entry with an $up$ or $left$ links hence,
 by following at most  two links we encounter an index needed to be
reported. Consequently, the time of following links is attributed
to the size of the output. \QED
\end{proof}

As the Lookup Query procedure can replace the intersection by range queries, which is executed $n(\beta - \alpha )$ times, Lemma \ref{l:queryLookTime} 
 proves the second part of Theorem~\ref{t:single}.

\subsection{Splitting the Text} \label{ss:split}
Usually, the input text is very long and arrives on-line. This makes the query algorithm requirement to insert all suffixes of the text unreasonable. Transforming this algorithm into an online algorithm seems a difficult problem. The main difficulty is working with on-line suffix trees construction in a sliding window. While useful constructions based on Ukkonen's~\cite{Ukk:95} and McCgright's~\cite{McC:76} suffix trees constructions exist (see~\cite{jaip:03}), no such results are known for Weiner's suffix tree construction, which our reversed prefixes tree construction depends on.

Nevertheless, we do not need to know the whole text in advance and we can process only separate chunks of it each time. To do this we take $m=\beta-\alpha+\max_i\sum_j|P_{i,j}|$ and split the text twice to pieces of size $2m$: first starting form the beginning of the text and the second starting after $m$ symbols. We then apply the algorithms for a single gap or $k$-gaps for each of the pieces separately for both text splits. Note that any appearance of a dictionary pattern can still be found by the algorithms on the splitted text.

\section{Conclusions and Open Problems}\label{s:open}
We showed that combinatorial string methods other than Aho-Corasick automaton can be applied to the $DMG$ problem to yield efficient algorithms. In this paper we focused on solving $DMG$, where a single gap exists in all patterns in the dictionary. We also relaxed the problem so that all patterns in the dictionary have the same gap bounds. It is an interesting open problem to study the general problem without these relaxations.

\end{document}